\documentclass[11pt]{article}
\usepackage{ifthen}
\usepackage{microtype}
\usepackage{boxedminipage}
\usepackage{graphicx}
\usepackage{amsmath}
\usepackage{amsfonts}
\usepackage{amssymb}
\usepackage{authblk}
\usepackage{hyperref}
\usepackage[english]{babel}
\usepackage{blindtext}
\usepackage{amssymb,amsfonts,epsfig}
    \newtheorem{theorem}{Theorem}
    \newtheorem{lemma}[theorem]{Lemma}
    \newtheorem{corollary}[theorem]{Corollary}
    \newtheorem{definition}[theorem]{Definition}
\newenvironment{proof}{\paragraph{Proof:}}{\hfill$\square$}

\newcommand{\pr}{{\bf Pr}}

\newcommand{\ignore}[1]{{}}

\newcommand{\commented}{no}

\ifthenelse{\equal{\commented}{no}}{
\newcommand{\jnote}[1]{}
\newcommand{\snote}[1]{}
\newcommand{\rnote}[1]{}
}

\title{Approximation Algorithms for Budget Constrained Network Upgradeable Problems }

\author{Debjyoti Saharoy and Sandeep Sen\\
Department of Computer Science and Engineering, \\
IIT Delhi, New Delhi 110016, India\thanks{Email addresses: {\{debjyoti,ssen\}@cse.iitd.ac.in}} 
}

\begin{document}

\maketitle
\begin{abstract}
We study \emph{budget constrained network upgradeable problems.} We are given 
an undirected edge weighted graph $G=(V,E)$ where the weight an edge 
$e \in E$ can be upgraded for a cost $c(e)$. 
Given a budget $B$ for improvement, the goal is to find a subset of edges
to be upgraded so that the resulting network is optimum for $B$. 
The results obtained in this paper include the following.
\begin{enumerate}
\item { \em Maximum Weight Constrained Spanning Tree}
\par We present a randomized algorithm for the problem of weight upgradeable 
budget constrained maximum spanning tree on a general graph. This returns 
a spanning tree 
$\mathcal{T}^{'}$ which is feasible within the budget $B$, such that
$\pr[ l(\mathcal{T}^{'}) \geq (1-\epsilon)\text{OPT}\text{ , } c(\mathcal{T}^{'}) \leq B] \ge 1-\frac{1}{n}$ 
(where $l$ and $c$ denote the length and cost of the tree respectively),
for any fixed $\epsilon >0$,
in time polynomial in $|V|=n$, $|E|=m$. Our results extend to the minimization
version also.
Previously Krumke et. al. \cite{krumke} 
presented a
$(1+\frac{1}{\gamma}, 1+ \gamma)$ bicriteria approximation algorithm for any 
fixed $\gamma >0$  for this problem in general graphs for a more general
cost upgrade function. The result in this paper improves their 0/1 cost 
upgrade model. 
\item {\em Longest Path in a DAG}
\par We consider the problem of weight improvable longest path in a $n$ vertex DAG and give a $O(n^3)$ algorithm for the problem when there is a bound on the number
of improvements allowed. We also give a $(1-\epsilon)$-approximation which runs in $O(\frac{n^4}{\epsilon})$ time for the budget constrained version.
Similar results can be achieved also for the problem of shortest paths in a DAG.
\end{enumerate}
\end{abstract}
\section{Introduction}
We consider optimization problems in budget constrained network 
{\it upgradeable} 
model, where the exact cost-benefit trade-offs are known\footnote{which may
not hold for an online model} a priori. 
In such optimization problems, there are two kinds of problem data: 
First, the nominal or unimproved data, and second, {\it improved} data of 
several stages that would allow solutions with better performance. 
We are allowed to use the improved data by incurring a cost within an 
overall budget for all the improvement cost.
Several computational problems fall in this category. 
For example, consider the 
problem of upgrading arcs in a network to minimize travel time \cite{NET:NET20097}. Another example would be a variant of min cost flow problem.
Here we are allowed to lower the flow cost of each arc,
and given a flow value and a bound on the total budget which can be used for lowering the flow costs, the goal is to find an upgrade strategy and a flow of minimum cost. This problem is considered by
\cite{Demgensky2002265}. In \cite{Krumke1999139} Krumke et. al. consider the problem of improving spanning trees by upgrading nodes. In communication networks 
upgrading could mean installing 
faster communication device at a node. 
Typically, such an upgradeable extension increases the complexity of an otherwise polynomially solvable problem to being NP-hard. The constrained minimum 
spanning tree problem subject to an additional linear constraint
was shown to {\it weakly}  NP Hard in \cite{Aggarwal1982287}. 

{\bf Remark} In the remaining paper, we will use the term {\it improvable} instead of
{\it upgradeable} as a common term to address both increase or
decrease in the parameter value, (as the case may be), to obtain a 
superior or improved objective value.

The {\em improvable} version of a problem is distinct from optimization
version that may have limited choices for certain kinds of objects.
For example, we may have multiple (parallel) edges for the network
problems offering trade-offs between cost and weight. This can be captured
within the framework of the original problem, using additional constraints,
and keeping the objective function unchanged.
The improvable version works on a budget constraint that is distinct from
the original constraints and there may not be any well-defined conversion 
between the budget constraint and the original constraints. For example,
in the constrained spanning tree problem in \cite{Aggarwal1982287}, the 
lengths and the weights are distinct parameters.
\par  Computing a minimum or maximum weight spanning tree for a network is a well studied problem in computer science.
Here we consider a variant of the maximization version of the problem,  
where the goal is to increase the edge lengths of a given network so that the length of a maximum
spanning tree in resulting network is as large as possible. The problem is considered in a context 
where there is a cost associated with improving an edge and there is a budget constraint on the total
cost of improving the edges. We denote the problem as $IMST$ for \emph{Improvable Maximum Spanning Tree}, the optimum tree as $\mathcal{T}^*$, its length as $l(\mathcal{T}^*)$ and the total cost of the higher length valued
edges as $c(\mathcal{T}^*).$ A precise definition of the problem is given in section \ref{sec:imst}. Since the maximum weight base problem on a matroid is equivalent to minimum cost base problem (by applying 
suitable transformations on weights) so our results and analysis hold 
even for weight improvable budget constrained minimum spanning tree on general graphs.
\par The second problem that we consider is the budget constrained 
network improvable longest or shortest path in a directed acyclic graph. 
Since the longest and shortest path problems in a DAG are related, we 
consider the former but similar
results and analysis would hold for the later as well. The longest path problem is the problem of 
finding a simple path of maximum length in a graph. The longest path
problem is NP-Hard and there is no polynomial time constant-factor
approximation algorithm unless P=NP (\cite{kar97}). However, the problem becomes 
polynomial time solvable on directed acyclic graphs. We investigate
this problem in budget constrained network improvable setting. 
We refer to the problem as $WILDAG$ for \emph{Weight-Improvable-Longest-Path-in-DAG}.
Longest path algorithms have applications in diverse fields.
The well known Travelling Salesman Problem is a special case of the
Longest Path Problem (\cite{1962}). The longest path in a program
activity graph is known as the critical path which represents the
sequence of program activities that take the longest time to execute. 
Longest path algorithms are required to calculate critical paths.
\subsection{Other Related Work}
More examples and applications of computational problems in the improvable framework can be found in \cite{impro-knap}. 
Goerigk, Sabharwal, Sch\"obel and Sen \cite{impro-knap} considered the 
weight-reducible knapsack problem, for which they gave a polynomial-time 
3-approximation and an 
FPTAS for the special case of uniform improvement costs. 
The problem of budget constrained network improvable spanning tree has been proved to 
be NP-hard, even for series-parallel graphs, by Krumke et. al. 
\cite{krumke}, which also cite several practical applications. 
Frederickson and 
Solis-Oba \cite{Fred} considered the problem of increasing 
the weight of the minimum spanning
tree in a graph subject to a budget constraint where the cost functions 
increase 
linearly with weights. 
Berman et. al. \cite{berman} consider the problem of shortening edges
in a given tree to minimize its shortest path tree weight. In contrast to 
most problems in the
network upgradation model, this problem was shown to be solvable in strongly polynomial time.
Phillips \cite{phil} studied the problem of finding an optimal strategy for reducing the capacity of the 
network so that the residual capacity in the modified network is minimized. 

Most of the network
modification results can be broadly classified as bi-criteria problems which are
characterized as $(\alpha , \beta )$ approximation if the algorithm achieves
factor $\alpha$ (respective $\beta$) approximation w.r.t. to the first 
(second) parameter.  
If $\alpha$ (or $\beta$) equals 1, then it is at least as good as the optimal
solution w.r.t. to the first (respectively second) parameter.
For the improvable spanning tree problem, the two parameters are the total
weight of the spanning tree and the budget available for improving the 
spanning tree edges. Krumke et al. \cite{krumke} show that Tree-width 
bounded graphs with linear reduction costs are  
$(1+\epsilon, 1+\xi)$-approximable for any fixed $\epsilon$, $\xi >0$. They 
also show general graphs 
are $(1+\frac{1}{\gamma}, 1+ \gamma)$-approximable for any fixed $\gamma >0$ 
which implies a trade-off between the two approximation factors (note that the 
balance occurs for $\gamma = 1$). In the more elaborate journal 
version, the authors \cite{Drangmeister} actually
consider three distinct models of improvements where the $0/1$ reduction comes
closest to this paper.
Ravi and Goemens \cite{ravigoemens96constrainedmst} have studied the constrained minimum spanning tree problem with two independent weight function
on the edges and 
gave a $(1,1+\epsilon)$ polynomial time approximation. 
Their method can be used to  derive a $(1+\epsilon,1)$ approximation for the 
constrained minimum spanning tree problem
that runs in pseudo-polynomial time (\cite{mad}).
\subsection{Our Contributions}
\par For the problem of \emph{IMST} which allows multistage improvements,  
we present a randomized algorithm that returns a spanning tree 
$\mathcal{T}^{'}$ which is feasible within the budget $B$ and has length at least $(1-\epsilon)$ times the OPT\footnote{The value of the optimum spanning tree
among those trees that use up a total improvement budget $B$}  with high probability, i.e
$\pr[ l(\mathcal{T}^{'}) \geq (1-\epsilon)\text{OPT}\text{ , } c(\mathcal{T}^{'}) \leq B] \ge 1-\frac{1}{n},$ for any fixed $\epsilon >0$,
in time polynomial in $|V|=n$, $|E|=m$. Our algorithm does not make any assumptions on the structure of the graph and works for general graphs. 
\par For the problem of \emph{WILDAG} we give 
an $O(n^3)$ algorithm for the special case of uniform improvement cost for each edge. We also consider the more general version with arbitrary 
improvement costs and a 
budget constraint on the total improvements and give a fully polynomial time 
approximation scheme for this problem. 
\par The primary observation that we exploit is that the optimal solution
comprises of the improved and unimproved versions in some ratio which is
not known but can be approximated in some way. For example, if we restrict
ourselves to the unweighted version (of the maximization problem), we can say that at least half of the
objective value is due to the improved or the unimproved elements. If we
know how to solve each version separately, even approximately, we can 
combine them to obtain a constant factor approximation overall. This idea
extends to the weighted version also. We also make use of
dynamic programming and scaling to obtain our results. Our techniques are
quite general and likely to be useful for other improvable problems.
\subsection{Organization of the paper}
\par In section \ref{sec:imst} we investigate the problem $IMST$. Beginning with a simple $\frac{1}{2}$-approximation for the special case
of uniform improvement costs in section \ref{subsec:pd1}, we consider 
more general improvement costs in section \ref{sec:wt_algo} and 
give a randomized 
bi-criteria $(1-\epsilon, 1)$-approximation. We extend the algorithm in 
\ref{sec:wt_algo} to handle multistage improvements in \ref{subsec:multilevel}.
\par In section \ref{sec:longest-DAG}, we consider the problem of longest 
path in DAG (\ref{sec:longest-DAG}) in the improvable framework. For a gentle 
exposition of the ideas and techniques, 
we first deal with the $WILDAG$ problem
for the special case of equal improvement cost for each edge and give an 
$O(n^3)$ algorithm. In section \ref{subsec:non-uni} we consider the general 
version with non-uniform improvement 
costs and a budget constraint and present 
a fully polynomial time approximation scheme in section \ref{sec:fptas}.
\section{Randomized $(1-\epsilon)$-approximation for $IMST$}
\label{sec:imst}
Given an undirected graph $G=(V,E)$, and nonnegative integers $l_e$ and $h_e$ 
(i.e each edge can be thought to have two copies with length $l_e$ and $h_e$), an improvement cost $c_e$, for each edge $e \in E$ ,and a budget $B$, we consider the problem of finding a 
spanning tree with maximum total edge length under the restriction that each time a higher valued edge is used, the associated improvement cost is 
incurred and the total improvement cost can be at most $B$. Similar to the
constrained MST problem, this variation is also computationally intractable. 
For completeness, we present a simple proof along the lines of
\cite{Drangmeister,Aggarwal1982287}.
\begin{theorem}
The Improvable MST (\emph{IMST}) problem is NP hard even when the graph is a tree.
\label{nph-1}
\end{theorem}
\begin{figure}
\begin{center}
\includegraphics[width=10cm]{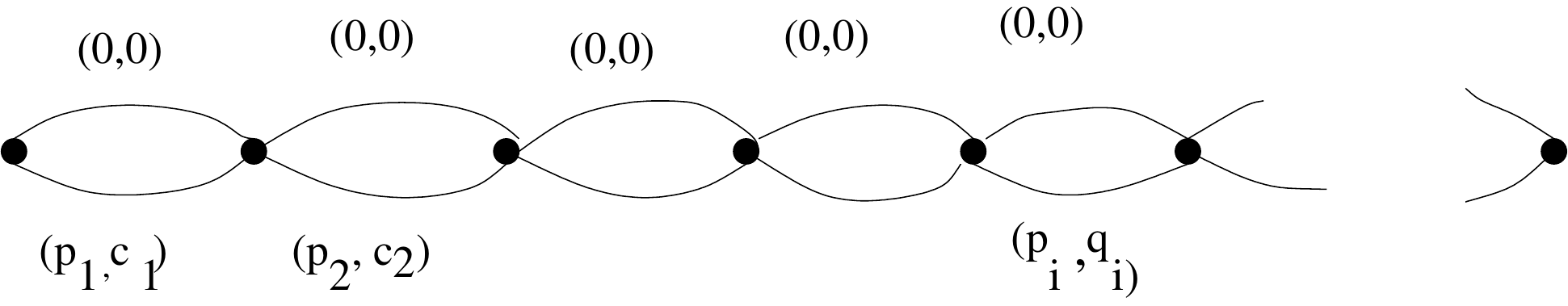}
\caption{Reducing Knapsack to IMST and WILDAG}
\end{center}
\label{fig:redn}
\end{figure}
\begin{proof}
We can reduce the knapsack problem using a construction similar to
\cite{Aggarwal1982287} - see Figure 1. The knapsack profit function
is $\max \sum_e p_e x_e $ and the constraint is
$ \sum_e c_e \cdot x_e \leq B $ 
where $B$ is also the budget constraint of the IMST problem. Wlog, we have
shown the two different versions of the edges as parallel edges, since
any spanning tree will include exactly one of the two edges. Note that
$l_e = 0$ and $h_e = p_e$. The IMST
solution will directly yield a solution to the knapsack problem.
\end{proof}
\subsection{Uniform improvement}
\label{subsec:pd1}
If the improvement costs are uniform, i.e., $c_e =1$, then the problem 
reduces to finding a spanning tree with maximum length subject to a cardinality constraint (say $k$) on the number of higher 
length valued edge used. Let us denote the problem as $UIMST$ 
for \emph{Uniform Improvable Maximum Spanning Tree}, the optimum tree as $\mathcal{T}^*$, and its length as $l(\mathcal{T}^*)$.  
\subsubsection{The Algorithm}
\label{subsec:unwt_algo_basic}
Note that the optimum tree is composed of two spanning forests one with edges having lengths $l_e$ and the other with lengths $h_e$ (at most $k$ of them). We denote the two forests
as $\mathcal{F}_{1}^{*}$ and $\mathcal{F}_{2}^{*}$ and their respective lengths as $ l(\mathcal{F}_{1}^{*}) $ and $ l(\mathcal{F}_{2}^{*}) $. Our algorithm is based
on finding two spanning trees $\mathcal{T}_{1}$ and $\mathcal{T}_{2}$ such that 
\begin{equation} l(\mathcal{T}_{1}) + l(\mathcal{T}_{2}) \geq l(\mathcal{F}_{1}^{*}) +l(\mathcal{F}_{2}^{*}) = l(\mathcal{T}^{*}) \nonumber
\end{equation} 
So that $\max{} {\{ l(\mathcal{T}_{1}), l(\mathcal{T}_{2})\}}  \geq \frac{1}{2} l(\mathcal{T}^{*}) $. Thus the better of the two trees $\mathcal{T}_{1}$, $\mathcal{T}_{2}$
is a $\frac{1}{2}$-approximate solution of \emph{UIMST}.
We can obtain $\mathcal{T}_{1}$ by invoking
Kruskal's algorithm for maximum spanning tree \cite{kruskal56algo} on the 
lower length valued edges, i.e $l_e$ for each $e \in E$. 
Clearly
$l(\mathcal{T}_{1}) \geq l(\mathcal{F}_{1}^{*})$. To find the tree $\mathcal{T}_{2}$ we first invoke Kruskal's algorithm \cite{kruskal56algo} on the high valued
edges, i.e $h_e$ for each $h_e \in E$, to obtain the maximum spanning forest in $G$ with at most $k$ edges where $k$ is the maximum number of improvements
allowed (we terminate the greedy algorithm after $k$ edges have been 
chosen). This forest can be extended to a tree $\mathcal{T}_{2}$ as 
spanning trees form the 
bases of a Graphic Matroid. Again, we have $l(\mathcal{T}_{2}) \geq l(\mathcal{F}_{2}^{*}).$ 
This completes the algorithm of section ~\ref{subsec:unwt_algo_basic}.
\subsection{Arbitrary improvement: An approximation algorithm}
\label{sec:wt_algo}
We now address the general case where $c_e$'s can be arbitrary. 
Let us denote the 
the maximal tree by $\mathcal{T}^*$, its length as $l(\mathcal{T}^*)$ and the total cost of the higher length valued
edges as $c(\mathcal{T}^*).$
\par Before we formally present the algorithm,
 we will need the following definition and a related result.
\begin{definition}{Two Cost Spanning Tree Problem:}
\label{def:2costmst}
Given a connected undirected graph $G = (V,E)$,
two edge weight functions, $c()$ and $l()$, and a bound $B$, find a 
spanning tree $\mathcal{T}^{*}$ of $G$ such that the total
cost $c( \mathcal{T}^*)$ is at most $B$ and the total cost $   l(  \mathcal{T}^* )$ is maximum among all spanning trees that stay within the
budget constraint.
\end{definition}

Observe that the definition is not completely
symmetric in the two weight functions; only the budget $B$ is specified. 
This problem has been addressed
earlier and the following bi-criteria approximation result was presented. 
\begin{theorem}[\cite{ravigoemens96constrainedmst}]
\label{thm:RGthm}
For all $\epsilon \ge 0 $, there is a polynomial time approximation algorithm for the Two Cost Spanning Tree problem with a performance of $(1, 1+\epsilon)$.  
\end{theorem} 
The result of Theorem \ref{thm:RGthm} holds even if the set of spanning trees 
are replaced by the bases of any matroid.
The reader may note that since the definition of the Two Cost Spanning Tree 
is not symmetric in the two weight functions, the $(1, 1+\epsilon)$ 
approximation does not imply  
a $(1+\epsilon,1)$ approximation for the constrained minimum spanning tree 
problem. Indeed the
the running time becomes pseudo polynomial (\cite{mad}).
\par Let $G'=(V,E')$ be a multi-graph obtained from $G$, where $E'$ is the 
set of edges that two copies of an edge $e \in G$ with weights 
$h_e$ and $l_e$ respectively.
The edges $h_e$ have cost $c_e$ and edges $l_e$ have costs $0$, for each 
$e \in E'$. 
 It is easy to see that multi-graph $G'$ retains the properties of a graphic
matroid. So we run the algorithm by \cite{ravigoemens96constrainedmst} on the multi-graph $G'$ with the budget constraint $B$. The output of
\cite{ravigoemens96constrainedmst} is a tree $\mathcal{T}$ such that $l(\mathcal{T}) \geq l(\mathcal{T}^{*}).$ Note that $\mathcal{T}$ is composed of
of a forest of high valued edges $h_e$, $\mathcal{F}_{1}$, and a forest of low 
valued edges $l_e$, $\mathcal{F}_{2}$. We must have $c(\mathcal{F}_{1}) 
\leq (1+ \epsilon) B$
and $c(\mathcal{F}_{2}) = 0$. To make the forest $\mathcal{F}_{1}$ feasible within the budget $B$, we sample the edges in $\mathcal{F}_{1}$ randomly
with probability $\frac{1}{(1+\epsilon)^2}$. That is, with probability,
 $\frac{1}{(1+\epsilon)^2}$ we retain the high valued edge $h_e \in \mathcal{F}_{1}$ otherwise
we pick the low valued copy of the edge i.e $l_e.$ We denote this randomly 
sampled subset of edges (which is a forest) $\mathcal{F}_{1}^{'}$. 
As $\mathbb{E}(l(\mathcal{F}_{1}^{'})) = \frac{1}{(1+\epsilon)^2} l(\mathcal{F}_{1})$ and $\mathbb{E}(c(\mathcal{F}_{1}^{'})) \leq \frac{1}{(1+\epsilon)}B$, we take the union of the two
forests $\mathcal{F}_{1}^{'}$ and $\mathcal{F}_{2}^{'}$ that forms a tree 
$\mathcal{T}^{'}$, since the edges are in 1-1 correspondence with 
$\mathcal{T}$. Then   
\begin{align*}
\mathbb{E}(l(\mathcal{T}^{'}))
& = \mathbb{E}(l(\mathcal{F}_{1}^{'})) + l(\mathcal{F}_{2}) \nonumber\\
& = \frac{1}{(1+\epsilon)^2} l(\mathcal{F}_{1}) + l(\mathcal{F}_{2}) \nonumber\\
& \geq \frac{1}{(1+\epsilon)^2} (l(\mathcal{F}_{1}) + l(\mathcal{F}_{2}))  = 
\frac{1}{(1+\epsilon)^2} l(\mathcal{T}) \nonumber\\
\end{align*}
Hence, we have 
\begin{align}
\mathbb{E}(l(\mathcal{T}^{'})) &\geq \frac{1}{(1+\epsilon)^2} l(\mathcal{T}^{*})   \label{eq:erl} \\
\mathbb{E}(c(\mathcal{T}^{'})) &\leq \frac{1}{(1+\epsilon)}B \label{eq:erl2} 
\end{align}
We have shown that our algorithm returns a spanning tree 
whose expected cost is feasible and has expected length at least
$\frac{1}{(1+\epsilon)^2}$ times the optimum.
We shall claim the following result.
\begin{theorem}
\label{thm:wtsp_app_thm}
Given an undirected graph $G=(V,E)$, an  error parameter $\epsilon$ and a confidence parameter $\delta$, the algorithm of section \ref{sec:wt_algo} returns a 
spanning tree $\mathcal{T}^{'}$ which is feasible within the budget $B$, such that \\
\begin{center}
$\pr[ l(\mathcal{T}^{'}) \geq (1-\epsilon)l(\mathcal{T}^{*}) \text{ , } c(\mathcal{T}^{'}) \leq B] \ge 1-\delta ,$
\end{center}
in time polynomial in $|V|=n$, $|E|=m$, $\frac{1}{\epsilon}$ and $\log(\frac{1}{\delta}).$
\end{theorem}
\begin{proof}
It suffices to achieve the above with $\delta=c$, where $c<1$ is any arbitrary constant. The error probability can be later boosted to the given $\delta$ 
by performing $O(\log \frac{1}{\delta})$ trials with error probability $c$ and taking the median.
\par Now to run the algorithm in section \ref{sec:wt_algo} let us choose $\epsilon'$ (note that $\epsilon'$ is the parameter of the algorithm in \cite{ravigoemens96constrainedmst} and $\epsilon$ is the desired error guarantee) such that $1-\epsilon = (1-\epsilon')^2$. Therefore, 
\begin{align*}
& l(\mathcal{T}^{'}) \leq (1-\epsilon) l(\mathcal{T}^{*}) \nonumber\\
& \Leftrightarrow l(\mathcal{T}^{'}) \leq (1-\epsilon')^2 \mu (1+\epsilon')^2 (\text{where }\mu = \mathbb{E}(l(\mathcal{T}^{'})))     \nonumber\\
& \Leftrightarrow l(\mathcal{T}^{'}) \leq (1-\epsilon'^{2})^2 \mu 
\leq (1-\epsilon'^{2}) \mu \nonumber\\
\end{align*}
The first implication follows from Equation \ref{eq:erl} in section \ref{sec:wt_algo}. Therefore,
\begin{align*}
\pr[l(\mathcal{T}^{'}) \leq (1-\epsilon) l(\mathcal{T}^{*})] 
& \leq \pr[l(\mathcal{T}^{'}) \leq (1-\epsilon'^{2}) \mu]  \nonumber\\
& \leq \boldsymbol{\exp}(-\frac{\epsilon'^{4}\mu}{3}) 
 \leq \boldsymbol{\exp}(-\frac{\epsilon'^{4}\frac{1}{(1+\epsilon')^2} l(\mathcal{T}^{*})}{3}) \nonumber\\
\end{align*}
The second inequality follows from Chernoff bounds and the third inequality from Equation \ref{eq:erl}.
If $\boldsymbol{\exp}(-\frac{\epsilon'^{4}\frac{1}{(1+\epsilon')^2} l(\mathcal{T}^{*})}{3}) 
< \boldsymbol{\exp}(-1)$ then $\pr[l(\mathcal{T}^{'}) \leq 
(1-\epsilon) l(\mathcal{T}^{*})] < \frac{1}{e}$.
If $l(\mathcal{T}^{*}) > \frac{3(1+\epsilon')^2 }{\epsilon'^4} = s $, 
we are done, otherwise, 
we can scale the edge lengths according to the following technical Lemma. 
\begin{lemma}
 \label{lemma:scaling}
The edge lengths of the graph can be additively scaled by $\frac{s}{n}$ 
while preserving the $(1-\epsilon)$ approximation guarantee. 
\end{lemma}
\begin{proof}
From our previous observation, we require 
$l(\mathcal{T}^{*})$ to be 
at least $\frac{3(1+\epsilon')^2 }{\epsilon'^4}=s$. So we add $\frac{s}{n}$ to 
each edge(both $l_e$ and $h_e$) of the graph. This 
ensures that $l(\mathcal{T}^{*})$ is at least $s$. Therefore, for the desired 
error guarantee $\epsilon$, we have,
\begin{align*}
& l(\mathcal{T}^{'}) + (s/n)n \geq (1-\epsilon) [l(\mathcal{T}^{*})+ (s/n)n] \nonumber\\
&\Leftrightarrow l(\mathcal{T}^{'}) + s \geq (1-\epsilon) (l(\mathcal{T}^{*})+ s)  \nonumber\\
&\Leftrightarrow l(\mathcal{T}^{'}) \geq (1-\epsilon)l(\mathcal{T}^{*}) - 
s\epsilon \geq (1-\epsilon)l(\mathcal{T}^{*}) - \epsilon l(\mathcal{T}^{*}) 
= (1-2\epsilon)l(\mathcal{T}^{*}) \nonumber\\
\end{align*} 
The last inequality follows from $l(\mathcal{T}^{*})\geq s$. 
By choosing $\epsilon' =\epsilon/2$ we obtain the desired error bound, i.e.,
\end{proof}
\begin{equation}
 \pr[l(\mathcal{T}^{'}) \leq (1-\epsilon' ) l(\mathcal{T}^{*})] = c_1 < \frac{1}{e} \label{eq:c1}
\end{equation}
\ignore{
Moreover,
\begin{align*}
& c(\mathcal{T}^{'}) \geq B \nonumber\\
& \Rightarrow c(\mathcal{T}^{'}) \geq (1+\epsilon')\mu \text{ where } \mu =
\mathbb{E}[ c  (\mathcal{T}^{'}) ]\nonumber\\ 
\end{align*}
where the implication follows from Equation \ref{eq:erl2}. \\
therefore, 
}
Using $\mu = \mathbb{E}[ c  (\mathcal{T}^{'})]$ and Equation \ref{eq:erl2}
\[ 
\pr[ c(\mathcal{T}^{'}) \geq B ] 
 \leq \pr[ c(\mathcal{T}^{'}) \geq (1+\epsilon')\mu ] 
\leq \boldsymbol{\exp}(-\frac{\epsilon'^{2}\mu}{3}) \nonumber\\
\]
The second inequality follows from Chernoff bounds. Again if 
$\boldsymbol{\exp}(-\frac{\epsilon'^{2}\mu}{3}) < \boldsymbol{\exp}(-1)$, then 
$\pr[ c(\mathcal{T}^{'}) \geq B ] < \frac{1}{e}$. 
Since $\mathbb{E}(c(\mathcal{T}^{'}))=\mu > \frac{3 }{\epsilon'^2}$ can be 
achieved by scaling the edge costs $c_e$ for each $h_e \in E$ along with the 
budget $B$. Note that $l_e$'s continue to have 0 weights 
unlike the construction in Lemma \ref{lemma:scaling}. However, the basic 
argument can be extended by considering only the edge set of $h_e$.
\begin{equation}
 \pr[ c(\mathcal{T}^{'}) \geq B ] = c_2 < \frac{1}{e} \label{eq:c2}.
\end{equation}
The polynomial running time follow from the PTAS 
for bi-criterion spanning tree given by \cite{ravigoemens96constrainedmst}.
From Equations \ref{eq:c1} and \ref{eq:c2}\\   
$\pr[ l(\mathcal{T}^{'}) \geq (1-\epsilon)l(\mathcal{T}^{*}) \text{ , } c(\mathcal{T}^{'}) \leq B] \ge 1-c, \text{ where }(c=c_1+c_2) $.
This concludes the proof of Theorem \ref{thm:wtsp_app_thm}.
\end{proof}

\begin{corollary}
\label{thm:wtsp_app_corr}
Given an undirected graph $G=(V,E)$ an error parameter $\epsilon$ the 
algorithm of section \ref{sec:wt_algo} returns a 
spanning tree $\mathcal{T}^{'}$ which is feasible within the budget $B$, such that \\
\begin{center}
$\pr[ l(\mathcal{T}^{'}) \geq (1-\epsilon)l(\mathcal{T}^{*}) \text{ , } c(\mathcal{T}^{'}) \leq B] \ge 1-\frac{1}{n} ,$
\end{center}
in time polynomial in $|V|=n$, $|E|=m$ and $\frac{1}{\epsilon}$. 
\end{corollary}
\subsection{Further extensions}
\label{subsec:multilevel}
We now consider a more general version of the problem, where
the edge lengths $l_e$ instead of admitting only a single improvement $h_e$, can have multiple
stages of improvement. That is for each $i \in [m]$, where $m$ is the number of edges, we are given
improved lengths $l_{i,1} \leq \cdots \leq l_{i,j(i)}$ with increasing improvement costs
$c_{i,1} \leq \cdots \leq c_{i,j(i)}$. The lengths $l_{i,0}$ has cost $c_{i,0} =0$ for each $i \in [m]$. The algorithm in section \ref{sec:wt_algo} can be extended to accommodate the multiple improvements in edge lengths. 
This is possible as each edge
$i \in [m]$ in the multi-graph, that we constructed in section \ref{sec:wt_algo}, can have $j(i)$ copies with respective costs $c_{i,1},\cdots,c_{i,j(i)}$. Again, \cite{ravigoemens96constrainedmst}
would build a tree composing of two forests of improved and non-improved edge lengths. And then the improved edge length forest can be sampled as in section \ref{sec:wt_algo} to ensure feasibility within the budget $B$.
Approximation guarantees similar to the single stage improvement 
can be proved along similar lines.

{\bf Remark} Our multistage improvement extend beyond the 0/1 upgradation
model of \cite{krumke} but it doesn't directly yield results for the
{\it integral} or the {\it rational} upgrade model of \cite{krumke}. However,
we can approximate the rational model to within any inverse polynomial by
choosing a polynomial number of parallel edges.  

Due to the equivalence of maximum weight base and minimum cost base over a 
matroid, the results and analysis done in this section would also hold for the minimization version. 

\ignore{
We shall now restate the theorem for weight improvable budget constrained minimum spanning tree.   
\begin{theorem}
\label{thm:min}
Given an undirected graph $G=(V,E)$, error parameter $\epsilon$ the algorithm of section \ref{sec:wt_algo} returns a 
spanning tree $\mathcal{T}^{'}$ which is feasible within the budget $B$, such that \\
\begin{center}
$\pr[ l(\mathcal{T}^{'}) \leq (1+\epsilon)l(\mathcal{T}^{*}) \text{ , } c(\mathcal{T}^{'}) \leq B] \ge 1-\frac{1}{n} ,$
\end{center}
in time polynomial in $|V|=n$, $|E|=m$ and $\frac{1}{\epsilon}$. 
\end{theorem}
}
\section{FPTAS for WILDAG}
\label{sec:longest-DAG}
Given a DAG, $G=(V, E)$, having $n$ vertices and $m$ edges
with edge lengths $l_{i},i \in [m]$ where 
$[m]$ $:=$ $\{1,...,m\}$. Let $LDAG(l,s,t)$
denote the problem of finding the longest $s-t$ path in $G$, for a given source
$s$ and a sink $t$. An instance of the \emph{Weight-Improvable-Longest-Path-in-DAG} is given by
the same set of vertices $[n]$, and edges $[m]$ with edge-lengths $l_{i},i \in [m]$,
source $s$ and sink $t$. Additionally, we are given the improved edge weights
$h_i$ for all $i \in [m]$, along with corresponding improvement costs $q_i$ and 
an improvement budget $B$. A problem instance of \emph{Weight-Improvable-Longest-Path-in-DAG}
is denoted by $WILDAG(l,h,q,B,s,t)$. The optimum solution is denoted by 
$WILDAG^{*}(l,h,q,B,s,t).$
\begin{theorem}
The Weight Improvable Longest path problem in DAG is NP-hard.
\label{np-hard2}
\end{theorem}
\begin{proof}
The Knapsack problem is also polynomial time reducible to the WILDAG problem
using a construction similar to Figure 1. The two directed edges have
associated tuples $( weight , cost )$ as $(0, 0)$ and $( p_e , c_e )$ 
respectively where the knapsack problem is given by 
\[ \max \sum_e p_e x_e \ \ s.t. \ \ \sum_e c_e \leq B \]
Here $B$ is the budget for improvement. 
\end{proof}
\subsection{A pseudo-polynomial algorithm for WILDAG }
\label{subsec:non-uni}
We can observe that if the query costs $q_i$ for all $i \in [m]$ are uniform, 
say q, 
then the query costs and improvement budget $B$ can be accordingly scaled, so 
that the improvement budget is basically the number of improvements allowed, say b.
Also note that since $G$ is a DAG, so the longest $s-t$ path can have at most
$n-1$ edges. Therefore $b \geq n-1$ is as good as infinite budget, so the interesting case
is only when $b<n-1.$ \\
We now give a dynamic programming formulation for $WILDAG$
with uniform query costs. Let $T=<v_1,.....,v_n>$ be a topological ordering of the
vertices of the $DAG$ $G$. Without loss of generality let us assume
$v_1=s$ and $v_n=t$. Let $L(v_i,q)$ denote the length of the longest $v_i-t$ 
path using the vertices $\{v_i,v_{i+1},....,t\}$ and 
at most $q$ improved edge-weights.
\begin{equation}
L(v_i,q)=\smash{\displaystyle\max_{ j: (i,j) \in E }} \smash{\displaystyle\max_{}} \{ L(v_j,q)+l_{e=(i,j)}, L(v_j,q-1) + h_{e=(i,j)} \} \label{eq:dp1} 
\end{equation}
This can be seen as follows. Let $S \subseteq T$ denote the set of vertices from
which the longest path to $t$ using at most $q$ improved edges is known. The invariant
for Equation \ref{eq:dp1} is that for each vertex $v \in S$, the length of the longest
path from $v$ to $t$ using at most $q$ improved edges is known. Thus the proof of correctness of \ref{eq:dp1} follows 
from the induction on size of $S$. Hence the longest path from $s$ to $t$ using at most $b$ improved edges is either by taking the improved edge $h_e$ from $s$ to $v_i$ and taking
the longest path from $v_i$ to $t$ using at most $b-1$ improved edges or by taking the un-improved edge $l_e$ from $s$ to $v_i$ and taking the longest path from $v_i$ to $t$ using 
at most $b$ improved edges, and maximizing over all neighbors $v_i$ of $s$.       
Use the base cases as $L(t,q)= 0$ for all $q \in [b] $, and $L(v_i,x) = -\infty $ for all $i \in [m] $ and $x<0$.\\
\ignore{
The algorithm is presented below in Figure~\ref{fig:uni-cost}, and we prove the correctness in \ref{subsubsec:uni-correct}.
\begin{figure}[h!]
{\footnotesize
\begin{center}
\fbox{
\begin{minipage}{\textwidth}\sf
\begin{tabbing}
xxx\=xxxx\=xxxx\=xxxx\=xxxx\=xxxx\= \kill
\>{\bf Algorithm} {\tt Uni-Cost-Imp-DAG}( DAG $G$, $T=<v_1=s,....,v_n=t>$, $b$, $\bar l$, $\bar h$ ) \ \ \ \ \\
\>{\bf Input}: \> \> - DAG $G$, \ \ \ \\
 \> \> \> - the topological ordering of the DAG $G$ i.e, $T=<v_1,....,v_n>$\\
\> \> \> - the improvement budget $b$ \\
\> \> \> - the matrix of original edge weights $\bar l = [l_1 ..... l_m]^T$ \\
\> \> \> - the matrix of improved edge weights $\bar h = [h_1 ..... h_m]^T$ \\
\> {\bf Output}: the length of the longest $s-t$ path using at most $b$ improved edge-weights \\
\\
\> Create the table $L = n \times n$; Set $L(t,q)= 0$ for all $q \in [b] $, $L(v_i,x) = -\infty $ for all $i \in [n] $ and $x<0$ ;\\
\> \textbf{for} $v_i=v_n,....,v_1$ \textbf{do} \\
\>\> \textbf{for} $q=0,....,b$ \textbf{do} \\
\>\> $L(v_i,q)=$ according to Equation \ref{eq:dp1}\\
\>\> \textbf{end for} \\
\> \textbf{end for} \\
\> \textbf{return} $L(s,b)$;
\end{tabbing}
\end{minipage}
}
\end{center}
}
\caption{A pseudo-polynomial time algorithm for WILDAG}
\label{fig:uni-cost}
\end{figure}
\subsubsection{Correctness and Running Time}
\label{subsubsec:uni-correct}
The correctness of the algorithm in Figure \ref{fig:uni-cost} follows from 
the correctness of recurrence Equation \ref{eq:dp1}. 
}
\begin{lemma}
The dynamic programming algorithm takes $O(n^3)$ time. 
\end{lemma}
\begin{proof}
 Each entry in the table can be computed in $O(n)$ time. The number of vertices in the DAG $G$ is $n$ and the number of improvements allowed is at most $n-2$. Hence the result follows.
\end{proof}
Now let us consider that the improvement costs $q_e$ for all $e\in E$ that may 
not be uniform, and the improvement budget is $B$. We propose the following
dynamic programming formulation. Let $L(v_i,w)$ denote the minimum improvement budget used such that the path from $v_i$ to $t$ using only the vertices 
$\{v_i,v_{i+1}....,t\}$ has length $w$. Let $W= \smash{\displaystyle\max_{ e \in E }}$ $h_e $, then $nW$ is a trivial upper bound on the length
of the longest $s-t$ path in the DAG $G$. The dynamic programming recurrence
follows along the lines of Equation \ref{eq:dp1}
\begin{equation}
L(v_i,w)=\smash{\displaystyle\min_{ j: (i,j) \in E }} \smash{\displaystyle\min_{}} \{ L(v_j,w-h_{e=(i,j)})+q_e, L(v_j,w-l_{e=(i,j)})\} \label{eq:dp2}  
\end{equation}
Use the base cases as $L(t,w)= \infty$ for all $w>0 $, $L(v_i,x) = \infty $ for all $i \in [n] $ and $x<0$, and $L(t,0)=0$. \\
The algorithm is presented below in Figure \ref{fig:non-uni-cost}, and we prove the correctness in \ref{subsubsec:non-uni-correct}.
\begin{figure}[h!]
{\footnotesize
\begin{center}
\fbox{
\begin{minipage}{\textwidth}\sf
\begin{tabbing}
xxx\=xxxx\=xxxx\=xxxx\=xxxx\=xxxx\= \kill
\>{\bf Algorithm} {\tt NonUni-Cost-Imp-DAG}( DAG $G$, $T=<v_1=s,....,v_n=t>$, $B$, $\bar q$ ,$\bar l$, $\bar h$ ) \ \ \ \ \\
\>{\bf Input}: \> \> - DAG $G$, \ \ \ \\
 \> \> \> - the topological ordering of the DAG $G$ i.e, $T=<v_1,....,v_n>$\\
\> \> \> - the improvement budget $B$ \\
\> \> \> - the matrix of the improvement costs $\bar q=[q_1 ..... q_m]^T $ \\ 
\> \> \> - the matrix of original edge weights $\bar l = [l_1 ..... l_m]^T$ \\
\> \> \> - the matrix of improved edge weights $\bar h = [h_1 ..... h_m]^T$ \\
\> {\bf Output}: the length of the longest $s-t$ path spending at most budget $B$ on improved \\
\> edge-weights \\
\\
\> Create the table $L = n \times nW $; $L(t,w)= \infty$ for all $w>0 $, $L(v_i,x) = \infty $ for all $i \in [n] $ and $x<0$,\\ 
\> and $L(t,0)=0$ ;\\
\> \textbf{for} $v_i=v_n,....,v_1$ \textbf{do} \\
\>\> \textbf{for} $w=0,....,nW$ \textbf{do} \\
\>\> $L(v_i,w)=$ according to Equation \ref{eq:dp2}\\
\>\> \textbf{end for} \\
\> \textbf{end for} \\
\> \textbf{return} $\smash{\displaystyle\arg max_{w}} \{ L(s,w) \leq B\} $; 
\end{tabbing}
\end{minipage}
}
\end{center}
}
\caption{A pseudo polynomial time algorithm for WILDAG with non uniform query costs}
\label{fig:non-uni-cost}
\end{figure}
\subsubsection{Correctness and Running Time}
\label{subsubsec:non-uni-correct}
The correctness of the algorithm in Figure \ref{fig:non-uni-cost} follows 
from the correctness of recurrence Equation \ref{eq:dp2}.
Let $S \subseteq T$ denote the set of vertices, such that the minimum improvement
cost incurred on paths of length $w$ from each vertex in $S$ to $t$ is known.
Maintaining this invariant and inducting over the size of $S$ gives the correctness of \ref{eq:dp2}.  
The minimum improvement budget used such that the path from $s$ to $t$ has length $w$ is either incurring a cost $q_e$ by taking the 
higher valued edge $h_e$ to $v_i$ and then incurring the minimum improvement budget such that the path from $v_i$ to $t$ (using vertices $\{v_{i},\cdots,t\}$ in the topological ordering $T$) has length
$w-h_e$, or not incurring any cost by taking the lower valued edge $l_e$ to $v_i$ and then incurring the minimum improvement budget such that the path from $v_i$ to $t$ 
has length $w-l_e$, and minimizing over all neighbors $v_i$ of $s$. So, the length
of the longest path form $s$ to $t$ is the largest $w$ over all $s-t$ paths of length $w$ such that the minimum improvement budget for the path is at most $B$.      
\begin{lemma}
The algorithm in Figure~\ref{fig:non-uni-cost} takes $O(n^3W)$ time. 
\end{lemma}
\begin{proof}
 Each entry in the table can be computed in $O(n)$ time. The number of 
vertices in the DAG $G$ is $n$ and the upper bound on the length of the longest $s-t$ path is $nW$.
Hence the result follows.
\end{proof}
{\bf Remark} For uniform cost, i.e., $q_e = 1$, we can rewrite the dynamic
programming in terms of budget to obtain a polynomial time algorithm. 
\subsection{A FPTAS for WILDAG }
\label{sec:fptas}
Next, by scaling the lengths,
we convert the previous pseudo polynomial time algorithm (in Figure~\ref{fig:non-uni-cost}) into
an efficient version by compromising with an approximation factor in the 
objective of attaining the longest length $s-t$ path. 
\ignore{
Suppose we want to compute
a solution with an objective value of at least $(1-\epsilon)WILDAG^{*}(l,h,q,B,s,t)$.
We use the scaling, for any edge $e \in E$, we consider its new lengths 
$l'_{e}= \left \lfloor{\frac{l_e}{K}}\right \rfloor $ and
$h'_{e}= \left \lfloor{\frac{h_e}{K}}\right \rfloor $, where
$K=\frac{\epsilon W}{n}$ and use this to run the dynamic programming in Figure~\ref{fig:non-uni-cost}.
We present the following algorithm  
\begin{figure}[h!]
{\footnotesize
\begin{center}
\fbox{
\begin{minipage}{\textwidth}\sf
\begin{tabbing}
xxx\=xxxx\=xxxx\=xxxx\=xxxx\=xxxx\= \kill
\>{\bf Algorithm} {\tt } \\
\>{\bf Input}: \> \> - A problem instance of WILDAG \ \ \ \\
 \> \> \> - $\epsilon > 0$\\

\>{\bf Output}:\> \> - the length of the path which is at least $(1-\epsilon)$ times the longest $s-t$ path \\
 \> \> \> - spending at most budget $B$ on improved edge-weights.\\
\\
\> Set $K=\frac{\epsilon W}{n}$. Let $l'_{e}= \left \lfloor{\frac{l_e}{K}}\right \rfloor $ and $h'_{e}= \left \lfloor{\frac{h_e}{K}}\right \rfloor $ for all $e \in E$ \\
\> Solve the instance WILDAG() using Algorithm in Figure~\ref{fig:non-uni-cost}. Let $L'$ be the resulting solution. \\
\> \textbf{return} $L'$ 
\end{tabbing}
\end{minipage}
}
\end{center}
}
\caption{A FPTAS for WILDAG with non uniform improvement costs}
\label{fig:fptas}
\end{figure}
}
Using $W=O(n/\epsilon)$, the running time of the algorithm in Figure~\ref{fig:non-uni-cost}  is $O(n^3.\frac{n}{\epsilon})$, which is similar to the classic FPTAS for Knapsack \cite{vaz}.
Thus the theorem follows.
\begin{theorem}
The dynamic programming algorithm for WILDAG (Figure~\ref{fig:non-uni-cost}) with non uniform improvement costs 
returns a solution with objective value at least $(1-\epsilon)WILDAG^*$ in 
$O(\frac{n^4}{\epsilon})$ time. 
\end{theorem}
{\bf Remark} In this section we have considered the problem of Weight Improvable Longest Path in a DAG but similar dynamic programming formulation extends to 
the problem of Weight Improvable Shortest Path in a DAG. 

\section{Concluding Remarks}
\ignore{
We have proposed a fully polynomial time approximation scheme for the problem of weight improvable longest path in a DAG when there is a budget constraint on the cost incurred for improvement. Our main contribution in this paper has been to give a randomized algorithm 
for the problem of weight improvable budget constrained maximum spanning tree, which returns a spanning tree 
$\mathcal{T}^{'}$ which is feasible within the budget $B$, such that
$\pr[ l(\mathcal{T}^{'}) \geq (1-\epsilon)\text{OPT}\text{ , } c(\mathcal{T}^{'}) \leq B] \ge 1-\frac{1}{n}$, for any fixed $\epsilon >0$,
in time polynomial in $|V|=n$, $|E|=m$. 
Our algorithm in section \ref{sec:wt_algo} (later extended in \ref{subsec:multilevel}), for the weight improvable budget constrained maximum spanning tree, handles discrete levels of improvements
on the edge weights.
}
It remains an open question if we can design a polynomial time algorithm for 
the uniform cost improvement version of the constrained MST.
Further, it would be interesting to extend our techniques so that we can
handle even continuous improvements on edge weights,
more specifically, like the rational update model in \cite{krumke}.
\\[0.1in]
{\bf Acknowledgement} The authors wish to thank Marc Goerigk for valuable
discussions and an anonymous reviewer for his suggestions on improving the
presentation of an earlier version of this paper.
{\small 
\bibliography{writeup}
}
\end{document}